\documentclass[11pt]{article}

\usepackage[T1]{fontenc}
\usepackage{lmodern}
\usepackage[margin=1in]{geometry}
\usepackage{amsmath,amssymb,amsthm,mathtools,mathrsfs}
\usepackage{bm}
\usepackage{enumitem}
\usepackage{microtype}
\usepackage{xcolor}
\usepackage{hyperref}

\hypersetup{
  colorlinks=true,
  linkcolor=blue!50!black,
  citecolor=blue!50!black,
  urlcolor=blue!50!black
}

\newtheorem{theorem}{Theorem}[section]
\newtheorem{proposition}[theorem]{Proposition}
\newtheorem{lemma}[theorem]{Lemma}
\newtheorem{corollary}[theorem]{Corollary}
\theoremstyle{definition}
\newtheorem{definition}[theorem]{Definition}

\DeclareMathOperator{\Herm}{Herm}

\newcommand{\C}{\mathbb C}

\newcommand{\cO}{\mathcal O}
\newcommand{\cA}{\mathcal A}
\newcommand{\cF}{\mathcal F}

\newcommand{\norm}[1]{\left\|{#1}\right\|_2}

\newcommand{\vx}{{\bm x}}

\title{\bfseries Phase Retrieval in $\mathbb C^4$ Requires Exactly \\ Eleven Measurements}
\author{Meng Huang\thanks{School of Mathematical Sciences, Beihang University, Beijing, China. Corresponding author. Email: \texttt{menghuang@buaa.edu.cn}}
\and Jihun Kim\thanks{Independent Researcher, Seoul, Republic of Korea. Email: \texttt{jihunkimkw@gmail.com}}}

\date{}

\begin{document}

\maketitle

\begin{abstract}
Determining the minimal number of intensity measurements required for phase retrieval in $\mathbb{C}^4$ has been a long-standing open problem. Prior to this work, the best-known results implied that this minimum was either $10$ or $11$. In this paper, we leverage characteristic classes and cohomology groups from differential topology to prove that no family of $10$ vectors in $\mathbb{C}^4$ possesses the phase retrieval property. Combining our lower bound with Vinzant's explicit eleven-vector construction establishes that the exact minimum is $11$. Our result yields a consequence for pure state quantum tomography, namely, a rank-one POVM on $\mathbb{C}^4$ requires exactly $11$ elements to be informationally complete for pure states. This further implies that three orthonormal bases are insufficient to uniquely distinguish all pure states in $\mathbb{C}^4$. Since four orthonormal bases are already known to be sufficient, we conclude that exactly four bases are required, thereby completely resolving the problem left in \cite{Carmeli15}.

\vspace{1em}
\noindent\textbf{Keywords:} phase retrieval; quantum tomography;  complex projective space;  normal bundle; Pontryagin classes.

%\vspace{0.5em}
%\noindent \textbf{2020 MSC:} 94A12, 42C15 (Primary); 57R20, 81P16 (Secondary).
\end{abstract}

\section{Introduction}

\subsection{The phase retrieval problem}
 Given a matrix $\bm{A}=(\bm{a}_1,\ldots,\bm{a}_N)\in\mathbb{C}^{M\times N}$, where $\bm{a}_j \in \mathbb{C}^M$ denote its columns, we define the associated intensity map as
\[
    \cF_{\bm{A}}:\mathbb{C}^M\longrightarrow\mathbb{R}^N,
    \qquad
    \cF_{\bm{A}}(\bm{x})
    =
    \bigl(\lvert \bm{a}_1^*\bm{x}\rvert^2,\ldots,\lvert \bm{a}_N^*\bm{x}\rvert^2\bigr),
\]
where $\bm{a}_j^*$ is the conjugate transpose of $\bm{a}_j$. The phase retrieval problem consists of recovering the unknown signal $\bm{x} \in \C^M$ from its intensity measurements $\cF_{\bm{A}}(\bm{x})$. The known vectors $\bm{a}_j$, $j=1,\ldots, N$,  are referred to as measurement vectors. Note that the intensity map is invariant under global phase:
\[
    \cF_{\bm{A}}(\omega\bm{x})=\cF_{\bm{A}}(\bm{x})
    \qquad \mbox{for all} \qquad  \omega\in\mathbb{T},
\]
where $\mathbb{T}=\{\omega\in\mathbb{C}:\lvert\omega\rvert=1\}$ denotes the complex unit circle. Therefore, the reconstruction question is whether $\cF_{\bm A}$ is injective modulo the global phase action of $\mathbb{T}$.

\begin{definition}\label{def:phase-retrieval}
A matrix $\bm{A} \in \mathbb{C}^{M\times N}$ is said to possess the phase retrieval property or is phase retrievable in $\mathbb{C}^M$ if
\[
    \cF_{\bm A}(\bm x)=\cF_{\bm A}(\bm y)
    \quad\Longrightarrow\quad
    \bm y=\omega\bm x
    \text{ for some }\omega\in\mathbb{T}.
\]
\end{definition}
The phase retrieval problem has a long history in harmonic analysis, algebraic geometry, and quantum tomography; see, for example, \cite{BCE06,BCMN14,CEHV15,F04,Heinosaari13,FSC05}. Furthermore, it is ubiquitous in many areas of the physical sciences and engineering, such as X-ray crystallography \cite{harrison1993phase,millane1990phase}, diffraction imaging \cite{shechtman2015phase,chai2010array}, microscopy \cite{miao2008extending}, astronomy \cite{fienup1987phase}, and optics and acoustics \cite{walther1963question}, where optical sensors and detectors are incapable of recording phase information.

A central question in phase retrieval asks for the minimal number of measurements $N$ such that there exists a matrix $\bm{A} \in \mathbb{C}^{M\times N}$ possessing the phase retrieval property in $\mathbb{C}^M$. In this paper, we focus on the case where $M=4$ and investigate the following problem:
 
\textbf{Problem 1:} \textit{Does there exist a family of ten vectors possessing the phase retrieval property in $\mathbb{C}^4$?}

 \subsection{Motivation and related work}
The motivation for studying this problem stems from the following three perspectives:

In the context of phase retrieval, a fundamental problem is determining the minimal measurement number $N$ for which the map $\cF_{\bm{A}}$ is injective on $\mathbb{C}^M$ up to a global phase. For this purpose, we use the notation $N_{\min}(\mathbb{C}^M)$ to denote the minimal $N$ for which the phase retrieval property holds, namely, 
\[
    N_{\min}(\mathbb{C}^M) := \min \left\{ N : \mbox{there exists a matrix $\bm{A} \in \mathbb{C}^{M\times N}$ such that $\cF_{\bm{A}}$ is injective on $\mathbb{C}^M/\mathbb{T}$} \right\}.
\]
Balan, Casazza, and Edidin \cite{BCE06} showed that $\cF_{\bm{A}}$ is injective on $\mathbb{C}^M/\mathbb{T}$ if $N \ge 4M-2$ and $\bm{a}_1,\ldots,\bm{a}_N$ are generic vectors in $\mathbb{C}^M$. This implies $N_{\min}(\mathbb{C}^M) \le 4M-2$. Later, Bandeira, Cahill, Mixon, and Nelson \cite{BCMN14} conjectured that $N_{\min}(\mathbb{C}^M) = 4M-4$. More precisely, they conjectured that (a) if $N < 4M - 4$, then $\cF_{\bm{A}}$ is not injective on $\mathbb{C}^M/\mathbb{T}$; and (b) if $N \geq 4M - 4$, then $\cF_{\bm{A}}$ is injective on $\mathbb{C}^M/\mathbb{T}$ for generic vectors $\bm{a}_j$, $j = 1, \dots, N$. Using tools from algebraic geometry, Conca, Edidin, Hering, and Vinzant \cite{CEHV15} proved that part (b) of the conjecture is valid, and part (a) is also true if $M = 2^k + 1$, $k \in \mathbb{N}_0$. This implies that $N_{\min}(\mathbb{C}^M) \le 4M-4$ for all $M$, and $N_{\min}(\mathbb{C}^M) = 4M-4$ if $M = 2^k + 1$. In particular, the exact minimum is known in dimensions $M=2,3,5,9,\ldots$.

On the other hand, Vinzant  \cite{Vinzant15} constructed an injective frame consisting of only $11 = 4M-5 < 4M-4$ vectors in $\mathbb{C}^4$, thereby disproving part (a) of this conjecture. This indicates that $4M-4$ is not minimal for certain dimensions. Topological nonembedding results of Heinosaari, Mazzarella, and Wolf \cite{Heinosaari13} give lower bounds for pure-state tomography and, after Parseval normalization, for standard phase retrieval. For $M>4$, Wang and Xu \cite{wangxu} proved a lower bound for generalized complex phase retrieval which, in particular, implies
\[
    N \ge 4M-2-2\alpha
\]
for standard phase retrieval, where $\alpha$ denotes the number of $1$'s in the binary expansion of $M-1$. For $M=4$, the lower bound $N_{\min}(\mathbb{C}^4) \ge 10$ follows from \cite{Heinosaari13} together with Parseval normalization. Combining this with Vinzant's eleven-vector construction gives
\[
    10 \le N_{\min}(\mathbb{C}^4) \le 11.
\]
Consequently, the exact minimum number of measurements in $\mathbb{C}^4$ has remained unresolved. This directly motivates our central question: does there exist a phase retrievable family of $10$ vectors in $\mathbb{C}^4$?

Our problem is also closely related to pure state quantum tomography, which is generally described by positive operator-valued measures (POVMs) \cite{Heinosaari13,WangShang2018}. More precisely, a pure state on $\mathbb{C}^{M}$ to be recovered is represented by a rank-one density operator $\bm{\rho}_{\bm{x}} = \bm{x}\bm{x}^*$ with $\norm{\vx} = 1$, and the corresponding probability distribution of the measurement outcomes is given by
\[
    \rho^{\bm E}(j) = \operatorname{tr} \bigl(\bm{\rho}_{\bm{x}} \bm{E}_j \bigr),  \quad j=1, \ldots, N,
\]
where $\bm{E}_j \in \mathbb{C}^{M\times M}$ are positive semidefinite matrices satisfying $\sum_{j=1}^N \bm{E}_j = I_M$. The POVM is called pure state informationally complete, if the resulting probability map is injective on the set of pure states. If every POVM element has rank one, we may write \(  \bm{E}_{j}  =
    \bm{\phi}_{j}\bm{\phi}_{j}^{*} \). 
In this case,
\[
     \rho^{\bm E}(j) 
    =
    \operatorname{Tr}
    \bigl(
        \bm{x}\bm{x}^{*}
        \bm{\phi}_{j}\bm{\phi}_{j}^{*}
    \bigr)
    =
    \left|
        \left\langle\bm{\phi}_{j},\bm{x}\right\rangle
    \right|^{2}.
\]
Thus, injectivity of a rank-one POVM on pure states is precisely the corresponding phase retrieval property.
%Hence, the problem of determining an unknown pure state from rank-one measurements $\bm{A}_j = \bm{\phi}_j\bm{\phi}_j^*$ is a special case of phase retrieval in which $\sum_{j=1}^N \bm{\phi}_j\bm{\phi}_j^* = I_M$. 
In $\C^4$, Heinosaari, Mazzarella, and Wolf \cite{Heinosaari13} give an eleven-outcome POVM that distinguishes all pure states, equivalently a construction using ten self-adjoint observables. Together with their lower bound, this leaves the minimum number of outcomes of a general pure-state informationally complete POVM at either ten or eleven. The construction is not restricted to rank-one POVM elements.
Consequently, these results do not settle whether ten rank-one measurements can distinguish all pure states. This distinction between general POVMs and rank-one POVMs is essential for the problem
considered here.

Furthermore, this question is tied to the problem where measurements are the rows of several  orthonormal bases. The number of orthonormal bases required to distinguish all pure quantum states is a well-studied topic in the quantum tomography literature \cite{Carmeli16,Carmeli15,SunLL}. Carmeli, Heinosaari, Schultz, and Toigo  showed that three bases suffice in dimension $2$, whereas four bases are necessary and sufficient in dimension $3$ and in all dimensions $M \geq 5$ \cite{Carmeli15}. For dimension $M=4$, the required number of bases is either three or four \cite{Carmeli15}, a question that has remained unresolved. Since distinguishing pure states with three orthonormal bases would mathematically imply the existence of a phase retrieval family consisting of at most ten vectors, a negative answer to Problem 1 would rule out the sufficiency of three orthonormal bases in $\mathbb{C}^4$. This would establish that exactly four orthonormal bases are required to distinguish all pure quantum states in $\mathbb{C}^4$, providing the third motivation for our interest in Problem 1.

Finally, we note a recent probabilistic result regarding the injectivity of phase retrieval with $N = 4M-5$ measurements. Specifically, Li \cite{Li2026} proved that for an $N \times M$ matrix with independent standard complex Gaussian entries, if $N = 4M-5$, the associated phase retrieval map is non-injective with positive probability. For dimension $M=4$, this corresponds to $N=11$ measurements. However, this probabilistic result neither asserts that every eleven-vector frame fails, nor does it settle the ten-vector question. 

For a broader context, a more in-depth account of the history of necessary and sufficient bounds for phase retrieval in $\mathbb{C}^M$ can be found in \cite{BCNT,GKR}. Furthermore, Bodmann and Hammen \cite{BH1,BH2} developed concrete algorithms and established error bounds for phase retrieval using frames with low redundancy. Additional phase retrieval algorithms under Gaussian random measurements with provable performance guarantees are discussed extensively in \cite{waldspurger2018phase,turstregion,cai2021solving,TTCai,ChenFan,Duchi,tan2019phase,huangsiam}. For a comprehensive overview of recent theoretical, algorithmic, and applied developments in phase retrieval, we refer the reader to the survey paper \cite{GKR}.

\subsection{Our contributions}
As stated previously, determining the minimal measurement number $N$ for which the map $\cF_{\bm{A}}$ is injective on $\mathbb{C}^4$ up to a global phase has remained an open problem. The best-known prior results could only constrain this value to $N_{\min}(\mathbb{C}^4) = 10$ or $N_{\min}(\mathbb{C}^4) = 11$. In this paper, we provide a negative answer to Problem 1, proving that no family of $10$ vectors in $\mathbb{C}^4$ possesses the phase retrieval property. Combining this lower bound with Vinzant's $11$-vector construction yields the exact value $N_{\min}(\mathbb{C}^4) = 11$. 

The proof is topological but crucially relies on information specific to rank-one intensity measurements. We show that if a ten-vector frame possesses the phase retrieval property, its associated intensity map naturally induces a smooth embedding
\[
    \mathbb{CP}^{3}\hookrightarrow\mathbb{R}^{9}.
\]
For every nonzero measurement vector, the zero locus of the associated intensity coordinate is a projective hyperplane $\mathbb{CP}^{2}\subset\mathbb{CP}^{3}$. The rank-one structure of that coordinate allows us to construct a nowhere-zero smooth section of the normal bundle over this hyperplane. However, the resulting geometric splitting of the restricted normal bundle is rigorously incompatible with its first Pontryagin class. This topological obstruction establishes our result by contradiction.

In quantum-tomographic language, our result demonstrates that three orthonormal bases cannot distinguish all pure states in $\mathbb{C}^{4}$. Since four orthonormal bases are already known to be sufficient, the minimum number of orthonormal bases required to distinguish every pair of pure states in $\mathbb{C}^{4}$ is exactly four.

Our main result is stated as follows:

\begin{theorem}\label{thm:main}
No family of ten vectors in $\mathbb{C}^4$ possesses the phase retrieval property.
\end{theorem}

Together with Vinzant's eleven-vector construction, Theorem~\ref{thm:main} implies the following result.

\begin{corollary}\label{cor:exact}
The minimum number of measurements required for phase retrieval in $\mathbb{C}^4$ is exactly $11$.
\end{corollary}

\begin{proof}
Theorem~\ref{thm:main} establishes the lower bound $N \geq 11$. Vinzant's explicit eleven-vector frame in $\mathbb{C}^4$ is known to perform phase retrieval \cite{Vinzant15}. Hence, the lower and upper bounds coincide, proving Corollary~\ref{cor:exact}.
\end{proof}

Our result has a direct consequence for the informational completeness of POVMs with respect to pure states using orthonormal bases in $\mathbb{C}^4$.

\begin{corollary} \label{cor:bases}
Four orthonormal bases are necessary and sufficient to distinguish all pure states in $\mathbb{C}^4$.
\end{corollary}

\begin{proof}
The sufficiency of four orthonormal bases in every finite dimension is known from the constructions in \cite{Carmeli15}. It remains to prove that three bases are insufficient. Let
\[
    \mathcal{B}_\ell=\{\bm{u}_{\ell,1},\bm{u}_{\ell,2},\bm{u}_{\ell,3},\bm{u}_{\ell,4}\},
    \qquad \ell=1,2,3,
\]
be three orthonormal bases of $\mathbb{C}^4$, and define $\bm{P}_{\ell,j}=\bm{u}_{\ell,j}\bm{u}_{\ell,j}^*$. 
Assume, for the sake of contradiction, that these three bases can distinguish all pure states. We will show that this implies the existence of a $10$-vector family possessing the phase retrieval property. Specifically, these $10$ vectors can simply be chosen as:
\[
    \bm{u}_{1,1},\bm{u}_{1,2},\bm{u}_{1,3},\bm{u}_{1,4}, \bm{u}_{2,1},\bm{u}_{2,2},\bm{u}_{2,3}, \bm{u}_{3,1},\bm{u}_{3,2},\bm{u}_{3,3}.
\]
To verify this, let $\bm{x},\bm{y}\in\mathbb{C}^4$ be two vectors yielding the same intensity measurements. Summing the four intensities corresponding to the first basis yields $\norm{\bm{x}}^2 = \norm{\bm{y}}^2$, which implies $\norm{\bm{x}} = \norm{\bm{y}}$. If this norm is zero, then $\bm{x}=\bm{y}=\bm{0}$. Otherwise, let $\tilde{\bm{x}}=\bm{x}/\norm{\bm{x}}$ and $\tilde{\bm{y}}=\bm{y}/\norm{\bm{y}}$. Note that $\sum_{j=1}^4\bm{P}_{\ell,j}=\bm{I}_4$ for $\ell=1,2,3$. Consequently, $\tilde{\bm{x}}\tilde{\bm{x}}^*$ and $\tilde{\bm{y}}\tilde{\bm{y}}^*$ are pure states that yield identical measurement outcomes with respect to the three orthonormal bases $\mathcal{B}_\ell$, $\ell=1,2,3$. By our assumption, this forces $\tilde{\bm{x}}\tilde{\bm{x}}^*=\tilde{\bm{y}}\tilde{\bm{y}}^*$, which means $\bm{x}=\bm{y} \pmod{\mathbb{T}}$. In either case, we obtain a set of ten vectors in $\mathbb{C}^4$ capable of phase retrieval, which directly contradicts Theorem~\ref{thm:main}. Thus, three bases are impossible, making four the minimal requirement.
\end{proof}

\subsection{Organization}
The remainder of this paper is organized as follows. In Section 2, we introduce the necessary notation and  concepts from differential geometry, specifically focusing on the properties of tangent and normal bundles over complex projective spaces. Section 3 is devoted to the proof of our main result, Theorem~\ref{thm:main}.  Finally, in Section 4, we conclude the paper with a brief discussion of our findings and outline potential directions for future research.

\section{Preliminaries}\label{sec:prelim}

\subsection{Notation and conventions}\label{sec:notation}
We use the following conventions throughout. Bold lowercase $\bm{x}$ and uppercase $\bm{X}$ letters denote vectors and matrices, respectively. The symbols $\mathbb{C}$, $\mathbb{R}$, and $\mathbb{Z}$ denote the complex numbers, real numbers, and integers. The complex projective space $\mathbb{CP}^3 = (\mathbb{C}^4\setminus\{0\})/\mathbb{C}^*$ is a manifold of complex dimension $3$. For a topological space $X$, $H^k(X;\mathbb{Z})$ denotes its $k$th integral singular cohomology group, and $H^*(X;\mathbb{Z}) = \bigoplus_{k\geq 0} H^k(X;\mathbb{Z})$ is the total cohomology ring equipped with the cup product.

Let $\mathcal{O}_X$ denote the  trivial complex line bundle. On $\mathbb{CP}^3$, $\mathcal{O}_{\mathbb{CP}^3}(-1)$ is the tautological line bundle, its dual $\mathcal{O}_{\mathbb{CP}^3}(1)$ is the hyperplane line bundle, and $\mathcal{O}_{\mathbb{CP}^3}(k) = \mathcal{O}_{\mathbb{CP}^3}(1)^{\otimes k}$. The direct sum of $r$ copies of a bundle $E$ is denoted by $E^{\oplus r}$, and the trivial real rank-$r$ bundle by $\underline{\mathbb{R}}^r$. The holomorphic tangent bundle $T^{1,0}\mathbb{CP}^3$ has complex rank $3$; its underlying real tangent bundle is $T\mathbb{CP}^3 \cong (T^{1,0}\mathbb{CP}^3)_{\mathbb{R}}$, where the subscript $\mathbb{R}$ indicates forgetting the complex structure. For a complex rank-$r$ bundle $E$, its total Chern class is $c(E) = \sum_{j=0}^r c_j(E)$. For a complex line bundle, its first Chern class $c_1$ coincides with the Euler class $e$ of its underlying oriented real two-plane bundle.

The hyperplane class $h = c_1(\mathcal{O}_{\mathbb{CP}^3}(1))$ generates $H^2(\mathbb{CP}^3;\mathbb{Z}) \cong \mathbb{Z}$. Geometrically, $h$ is the Poincar\'e dual of the inclusion of a projective hyperplane $\iota: \mathbb{CP}^2 \hookrightarrow \mathbb{CP}^3$. By duality, $c_1(\mathcal{O}_{\mathbb{CP}^3}(-1)) = -h$. Higher-degree classes are formed by cup products,  e.g., $h^2 \in H^4, h^3 \in H^6$. Since the real dimension of $\mathbb{CP}^3$ is $6$, $H^8(\mathbb{CP}^3;\mathbb{Z}) = 0$, which necessarily forces $h^4 = 0$. The total cohomology ring is thus completely determined as
\begin{equation}\label{eq:cohomology-ring}
    H^*(\mathbb{CP}^3;\mathbb{Z}) \cong \mathbb{Z}[h]/(h^4), \qquad \deg h = 2,
\end{equation}
where the quotient of the polynomial ring $\mathbb{Z}[h]$ by the ideal $(h^4)$ algebraically enforces this geometric truncation. Pulling back to the hyperplane, we denote $\bar{h} = \iota^*h$. It immediately follows that $H^*(\mathbb{CP}^2;\mathbb{Z}) \cong \mathbb{Z}[\bar{h}]/(\bar{h}^3)$, where $\bar{h}^2$ generates $H^4(\mathbb{CP}^2;\mathbb{Z}) \cong \mathbb{Z}$. For the standard properties of Chern, Pontryagin, and Euler classes used below, we refer to \cite{MilnorStasheff}.

\subsection{Basic differential geometry: tangent and normal bundles}
We briefly recall the differential-geometric concepts utilized in our proof. If $M$ is a smooth real manifold and $p\in M$, the tangent space $T_pM$ is the real vector space consisting of the velocity vectors of smooth curves passing through $p$. More precisely, two curves $\gamma_1,\gamma_2:(-\varepsilon,\varepsilon)\to M$ with $\gamma_1(0)=\gamma_2(0)=p$ represent the same tangent vector if their first-order derivatives agree in one, and consequently in all, smooth local coordinate charts at $p$. The disjoint union
\[
    TM=\bigsqcup_{p\in M}T_pM
\]
forms the tangent bundle of $M$. If $M$ has real dimension $n$, then $TM$ is a real vector bundle of rank $n$.
For a smooth map $f:M\to N$, its differential at $p$ is the linear map
\[
    df_p:T_pM\longrightarrow T_{f(p)}N
\]
induced by the derivatives of these smooth curves. The map $f$ is an immersion if the differential $df_p$ is injective at every point $p \in M$. It is an embedding if it is an immersion and also a homeomorphism from $M$ onto its image, with the image carrying the induced subspace topology. For our purposes, the compactness of $\mathbb{CP}^3$ guarantees that any injective immersion into a Euclidean space is automatically a smooth embedding.

Suppose now that $f:M\hookrightarrow\mathbb{R}^q$ is a smooth embedding of an $n$-dimensional manifold. The standard Euclidean inner product defines, at each $p\in M$, the normal space
\[
    \nu_{f,p}
    =\bigl(df_p(T_pM)\bigr)^\perp
    \subset T_{f(p)}\mathbb{R}^q\cong\mathbb{R}^q.
\]
These spaces collectively form the normal bundle
\[
    \nu_f=\bigsqcup_{p\in M}\nu_{f,p}\longrightarrow M,
\]
which has a real rank of $q-n$. A section of $\nu_f$ is therefore a smooth assignment of a normal vector to the embedded manifold at every point. The orthogonal decomposition of the ambient tangent bundle yields the bundle isomorphism
\[
    f^*T\mathbb{R}^q\cong TM\oplus\nu_f.
\]
Because the tangent bundle $T\mathbb{R}^q$ is trivial, this reduces to the stable tangent-normal relation
\[
    TM\oplus\nu_f\cong\underline{\mathbb{R}}^{\,q}.
\]
If both the ambient space and $M$ are oriented, this relation induces a canonical orientation on $\nu_f$.

If $H\subset M$ is a submanifold and $s$ is a nowhere-zero smooth section of $\nu_f|_H$, then the line spanned by $s(p)$ constitutes a trivial real line subbundle. Taking its orthogonal complement yields the Whitney sum splitting
\[
    \nu_f|_H\cong\underline{\mathbb{R}}\oplus\eta,
\]
where $\eta$ is a vector bundle of rank $\operatorname{rank}(\nu_f)-1$. In our proof stated below, $M=\mathbb{CP}^3$, $q=9$, and $H\cong\mathbb{CP}^2$; hence, $\nu_f$ has rank three, and the nowhere-zero normal section produces an oriented rank-two bundle $\eta$. This provides the geometric justification for equating the first Pontryagin class of $\nu_f|_H$ with the square of the Euler class of $\eta$. For comprehensive background on smooth manifolds, tangent bundles, and normal bundles, we refer the reader to \cite{Lee09}.

\section{Proof of the main result}
\subsection{Overview of the proof}
The proof consists of four main ingredients. First, in Lemma~\ref{lem:normalize}, we demonstrate that it suffices to consider the case where the measurement vectors form a Parseval frame. Second, adapting a result from \cite[Lemma 9]{BCMN14}, we restate in Lemma~\ref{lem:kernel} that a Parseval frame possesses the phase retrieval property if and only if there are no nonzero Hermitian matrices of rank at most two in the kernel of its associated linear measurement map. 
Next, we show that if a Parseval frame consisting of ten vectors possesses the phase retrieval property, it canonically defines a smooth embedding
\[
    \mathbb{CP}^3\hookrightarrow\mathbb{R}^9.
\]
The normal bundle of this embedding has a real rank of three. For any nonzero measurement vector $\bm{b}_j$, exploiting the rank-one structure of the coordinates allows us to show that on the hyperplane
\[
    H_j=\{[\bm{x}]\in\mathbb{CP}^3:\bm{b}_j^*\bm{x}=0\}\cong\mathbb{CP}^2,
\]
the restricted normal bundle splits off a trivial real line. Finally, utilizing characteristic classes, we prove that such a topological splitting is rigorously impossible. This contradiction establishes our main result.

\subsection{Auxiliary lemmas}
We begin with a straightforward but essential lemma regarding frame normalization.

\begin{lemma}\label{lem:normalize}
Suppose that the matrix $\bm{A}=(\bm{a}_1,\ldots,\bm{a}_N)\in\mathbb{C}^{M\times N}$ possesses the phase retrieval property in $\mathbb{C}^M$. Then its column vectors span $\mathbb{C}^M$, and the associated frame operator
\[
    \bm{S} := \sum_{j=1}^{N}\bm{a}_j\bm{a}_j^*
\]
is strictly positive definite. By defining $\bm{b}_j := \bm{S}^{-1/2}\bm{a}_j$ for $j=1,\ldots,N$, the resulting normalized frame $\bm{B}=(\bm{b}_1,\ldots,\bm{b}_N)$ is a Parseval frame satisfying
\[
    \sum_{j=1}^{N}\bm{b}_j\bm{b}_j^*=\bm{I}_M,
\]
and $\bm{B}$ has the phase retrieval property.
\end{lemma}
\begin{proof}
If the vectors $\{\bm a_j\}_{j=1}^N$ did not span $\mathbb C^M$, any nonzero vector $\bm v$ in their orthogonal complement would yield $|\bm a_j^*\bm v|^2 = 0$ for all $j$. This matches the intensity measurements of the zero vector, which contradicts the injectivity required for phase retrieval since $\bm v \neq \bm 0$. Hence, $\bm S$ is strictly positive definite. By construction, we verify the Parseval condition:
\[
    \sum_{j=1}^{N}\bm b_j\bm b_j^*
    =
    \bm S^{-1/2}\left(\sum_{j=1}^{N}\bm a_j\bm a_j^*\right)\bm S^{-1/2}
    =
    \bm I_M.
\]
For any $\bm x\in\mathbb C^M$, letting $\bm y=\bm S^{1/2}\bm x$ yields
\[
    \bm b_j^*\bm y
    =
    \bm a_j^*\bm S^{-1/2}\bm S^{1/2}\bm x
    =
    \bm a_j^*\bm x.
\]
Thus, the intensity maps of $\bm B$ evaluated at $\bm y$ are identical to those of $\bm A$ evaluated at $\bm x$. Since they are related by the invertible linear transformation $\bm y=\bm S^{1/2}\bm x$, which inherently preserves equivalence modulo a global phase, $\bm B$ strictly inherits the phase retrieval property.
\end{proof}

For a family matrix $\bm B=(\bm b_1,\ldots,\bm b_N)\in\mathbb C^{M\times N}$, define the rank-one operators
$\bm P_j=\bm b_j\bm b_j^*$ and the real-linear measurement map $\cA_{\bm B}:\Herm(M)\longrightarrow\mathbb R^N$ given by
\[
    \cA_{\bm B}(\bm X)
    =
    \bigl(\operatorname{Tr}(\bm P_1\bm X),\ldots,\operatorname{Tr}(\bm P_N\bm X)\bigr)
    =
    \bigl(\bm b_1^*\bm X\bm b_1,\ldots,\bm b_N^*\bm X\bm b_N\bigr).
\]
Here, $\Herm(M)$ denotes the space of $M\times M$ Hermitian matrices, treated as a real vector space.

The subsequent lemma is a restatement of Lemma 9 from \cite{BCMN14}. We include a self-contained proof here for completeness.
\begin{lemma} \label{lem:kernel}
Let $\bm B=(\bm b_1,\ldots,\bm b_N)$ be a Parseval frame in $\mathbb C^M$, namely, $\sum_{j=1}^{N}\bm b_j\bm b_j^*=\bm I_M$. Then
$\bm B$ has the phase retrieval property if and only if
\begin{equation} \label{eq:HermKer}
    \ker\cA_{\bm B}
    \cap
    \{\bm X\in\Herm(M):\operatorname{rank}\bm X\leq2\}
    =
    \{\bm 0\}.
\end{equation}
\end{lemma}

\begin{proof}
(\textit{Sufficiency}) We prove it by contradiction. Suppose that $\bm B$ fails to have the phase retrieval property. Then there exist $\bm x,\bm y\in\mathbb C^M$ that are not identical up to a global phase, yet yield the same measurements:
\[
    |\bm b_j^*\bm x|^2=|\bm b_j^*\bm y|^2
    \qquad \text{for all } j=1,\ldots,N.
\]
Construct the matrix $\bm X=\bm x\bm x^*-\bm y\bm y^*$. If $\bm X=\bm 0$, it would follow that $\bm x\bm x^*=\bm y\bm y^*$, forcing $\bm x$ and $\bm y$ to be related by a global phase, which contradicts our assumption. Thus, $\bm X$ is a nonzero Hermitian matrix with $\operatorname{rank}\bm X\leq2$. Moreover,
\[
    \operatorname{Tr}(\bm P_j\bm X)
    =
    |\bm b_j^*\bm x|^2-|\bm b_j^*\bm y|^2
    =
    0 \qquad \mbox{for all} \quad j,
\]
 implying $\bm X \in \ker\cA_{\bm B}$. This contradicts \eqref{eq:HermKer}.

(\textit{Necessity}) Conversely, assume there exists a nonzero Hermitian matrix $\bm X\in\ker\cA_{\bm B}$ with $\operatorname{rank}\bm X\leq2$. Since $\bm B$ is a Parseval frame, the trace of $\bm X$ vanishes:
\[
    \operatorname{Tr}(\bm X)
    =
    \operatorname{Tr}\!\left(\sum_{j=1}^{N}\bm P_j\bm X\right)
    =
    \sum_{j=1}^{N} \bm b_j^* \bm X \bm b_j
    =
    0.
\]
Because $\bm X \neq \bm 0$ has rank at most two and zero trace, it must have exactly rank two, with nonzero eigenvalues $\lambda$ and $-\lambda$ for some $\lambda>0$. By the spectral theorem, $\bm X$ admits the eigendecomposition
\[
    \bm X=\lambda(\bm u\bm u^*-\bm v\bm v^*)
\]
for some orthonormal vectors $\bm u,\bm v$. Setting $\bm x=\sqrt{\lambda}\,\bm u$ and $\bm y=\sqrt{\lambda}\,\bm v$, we have $\bm X=\bm x\bm x^*-\bm y\bm y^*$. Since $\bm u$ and $\bm v$ are nonzero and orthogonal, $\bm x$ and $\bm y$ are likewise orthogonal and thus cannot be related by a global phase. Finally, since $\bm X\in\ker\cA_{\bm B}$, we have
\[
    |\bm b_j^*\bm x|^2-|\bm b_j^*\bm y|^2
    =
    \operatorname{Tr}(\bm P_j\bm X)
    =
    0
    \qquad \text{for all } j=1,\ldots,N.
\]
This demonstrates that $\bm x$ and $\bm y$ have identical measurements, meaning phase retrieval fails.
\end{proof}

\subsection{Proof of Theorem~\ref{thm:main}}
Assume, for the sake of contradiction, that the matrix
\[
    \bm{A}=(\bm{a}_1,\ldots,\bm{a}_{10})\in\mathbb{C}^{4\times 10}
\]
possesses the phase retrieval property. By Lemma~\ref{lem:normalize}, we may assume without loss of generality that the measurement vectors form a Parseval frame, denoted by
\[
    \bm{B}=(\bm{b}_1,\ldots,\bm{b}_{10}),
    \qquad
    \sum_{j=1}^{10}\bm{b}_j\bm{b}_j^*=\bm{I}_4,
\]
which also has the phase retrieval property.
Let $S^{7}=\{\bm{x}\in\mathbb{C}^4:\|\bm{x}\|=1\}$ denote the unit sphere in $\mathbb{C}^4$. Every nonzero vector $\bm{x}\in\mathbb{C}^4$ can be expressed uniquely, up to a global phase, as $\bm{x}=\|\bm{x}\|\bm{u}$ with $\bm{u}\in S^7$. The quotient of $S^7$ by the free action of the unit circle $\mathbb{T}$ yields the complex projective space:
\[
    S^7/\mathbb{T} = \mathbb{CP}^3.
\]
We regard $\mathbb{CP}^3$ as a real smooth manifold of dimension $6$. Since the measurement vectors form a Parseval frame, we can restrict our attention to unit vectors, allowing the phase retrieval property to naturally induce a well-defined map on $\mathbb{CP}^3$.
For convenience, we define the projective intensity coordinates as
\[
    \rho_j([\bm{x}]):=\lvert \bm{b}_j^*\bm{x}\rvert^2,
    \qquad [\bm{x}]\in\mathbb{CP}^3,
\]
where $\bm{x}\in S^7$ is any unit representative of the equivalence class $[\bm{x}]$. This expression is well-defined because replacing $\bm{x}$ with $\omega\bm{x}$, $\omega\in\mathbb{T}$ leaves the absolute value invariant.
Now, let
\[
    \mathcal{Y}=\left\{\bm{y}=(y_1,\ldots,y_{10})\in\mathbb{R}^{10}:
    \sum_{j=1}^{10}y_j=1\right\},
    \qquad
    \mathbb{E}_0=\left\{\bm{u}\in\mathbb{R}^{10}:
    \sum_{j=1}^{10}u_j=0\right\},
\]
and define the mapping
\begin{equation} \label{eq:mapF}
    F:\mathbb{CP}^3\longrightarrow\mathcal{Y},
    \qquad
    F([\bm{x}])=\bigl(\rho_1([\bm{x}]),\ldots,\rho_{10}([\bm{x}])\bigr).
\end{equation}
The Parseval identity guarantees that $\sum_{j=1}^{10}\rho_j([\bm{x}])=1$, confirming that the image of $F$ indeed lies within the affine hyperplane $\mathcal{Y}$. The following proposition establishes that $F$ is a smooth map.

\begin{proposition}\label{prop:smooth-F}
Let $\bm{B}=(\bm{b}_1,\ldots,\bm{b}_{10})\in\mathbb{C}^{4\times 10}$ be a Parseval frame, and let $F$ be defined as in \eqref{eq:mapF}. Then $F$ is a smooth map between real smooth manifolds.  
\end{proposition}

\begin{proof}
Let $q:S^7\longrightarrow \mathbb{CP}^3$ be the canonical quotient map, and define the lifted intensity map
\[
    \widetilde{F}:S^7\longrightarrow\mathbb{R}^{10},
    \qquad
    \widetilde{F}(\bm{x})
    =
    \bigl(|\bm{b}_1^*\bm{x}|^2,\ldots,|\bm{b}_{10}^*\bm{x}|^2\bigr).
\]
Since each coordinate function $\bm{x}\mapsto |\bm{b}_j^*\bm{x}|^2$ is a real homogeneous polynomial of degree two with respect to the real and imaginary parts of $\bm{x}$, the map $\widetilde{F}$ is smooth. Furthermore, $\widetilde{F}$ satisfies
\[
    \widetilde{F}(\omega\bm{x})=\widetilde{F}(\bm{x})
    \qquad \text{for all } \omega\in\mathbb{T},
\]
meaning that $\widetilde{F}$ is $\mathbb{T}$-invariant.  Because the projection $q:S^7\to\mathbb{CP}^3$ constitutes a smooth principal $\mathbb{T}$-bundle, the characteristic property of surjective smooth submersions \cite[Theorem 4.29]{Lee12} dictates that any smooth, $\mathbb{T}$-invariant map on the total space $S^7$ descends to a unique smooth map on the base space $\mathbb{CP}^3$. 
Specifically, there exists a unique smooth map $F_0:\mathbb{CP}^3\longrightarrow\mathbb{R}^{10}$ such that $\widetilde{F}=F_0\circ q$. Comparing this with \eqref{eq:mapF}, we observe that $F(q(\bm{x})) = F([\bm{x}]) = \widetilde{F}(\bm{x})$ for all $\bm{x}\in S^7$, which implies $F \circ q = \widetilde{F}$. The surjectivity of $q$ then forces $F = F_0$. Consequently, $F$ is smooth as a map into $\mathbb{R}^{10}$. Finally, because $\mathcal{Y}$ is a smoothly embedded submanifold of $\mathbb{R}^{10}$ and the image of $F$ lies entirely within $\mathcal{Y}$, $F$ is also smooth as a map into $\mathcal{Y}$.
\end{proof}

Let $\bm{1}=(1,\ldots,1)\in\mathbb{R}^{10}$ denote the all-ones vector, and define the centered map
\begin{equation}\label{eq:centered-map}
    G:\mathbb{CP}^3\longrightarrow\mathbb{E}_0,
    \qquad
    G([\bm{x}])=F([\bm{x}])-\frac{1}{10}\bm{1}.
\end{equation}
Since $G$ is formed by translating $F$ by a constant vector, they share the identical differential. Furthermore, because $\mathbb{E}_0$ is a nine-dimensional Euclidean vector space, we can naturally regard $G$ as a map into $\mathbb{R}^9$. We now establish that $G$ is a smooth embedding.

\begin{proposition}\label{prop:projective-embedding}
Let $\bm{B}=(\bm{b}_1,\ldots,\bm{b}_{10})\in\mathbb{C}^{4\times 10}$ be a Parseval frame possessing the phase retrieval property, and let $G$ be defined as in \eqref{eq:centered-map}. Then $G$ is a smooth embedding.
\end{proposition}

\begin{proof}
The phase retrieval property inherently guarantees that $F$, and consequently $G$, is injective on $\mathbb{CP}^3$.   Because $\mathbb{CP}^3$ is compact and $\mathbb{E}_0 \cong \mathbb{R}^9$ is Hausdorff, any smooth, injective immersion between them is necessarily a smooth embedding. 
Thus, it suffices to prove that $G$ is a smooth immersion.

Fix a point $[\bm{x}]\in\mathbb{CP}^3$ and choose a unit representative $\bm{x}\in S^7$. The differential of the quotient map $q: S^7 \to \mathbb{CP}^3$ canonically identifies the horizontal subspace $\bm{x}^\perp_{\mathbb{C}}$ with the real tangent space $T_{[\bm{x}]}\mathbb{CP}^3$. Therefore, any tangent vector $\bm{v} \in T_{[\bm{x}]}\mathbb{CP}^3$ can be uniquely represented as $\bm{v} = dq_{\bm{x}}(\bm{z})$ for some horizontal vector
\begin{equation} \label{eq:tangent}
    \bm{z}\in \bm{x}^\perp_{\mathbb{C}}
    =
    \{\bm{z}\in\mathbb{C}^4:\bm{x}^*\bm{z}=0\}.
\end{equation}
Construct the associated Hermitian matrix
\[
    \bm{X}=\bm{x}\bm{z}^*+\bm{z}\bm{x}^*.
\]
It is straightforward to verify that $\bm{X}$ satisfies
\[
    \operatorname{rank}\bm{X}\leq 2, \qquad
    \bm{X}=\bm{X}^*, \qquad
    \operatorname{Tr}(\bm{X})=\bm{z}^*\bm{x}+\bm{x}^*\bm{z}=0.
\]
Recall from the definition of $\cA_{\bm{B}}$ that $\bm{P}_j=\bm{b}_j\bm{b}_j^*\in\mathbb{C}^{4\times 4}$. Let $\bm{x}(t)\in S^7$ be a smooth curve such that $\bm{x}(0)=\bm{x}$ and $\bm{x}'(0)=\bm{z}$. Evaluating the differential of the intensity coordinate $\rho_j$ along the projected curve $q(\bm{x}(t))$ yields
\[
d\rho_j\big|_{[\bm{x}]}(\bm{v}) =
\frac{d}{dt}\bigg|_{t=0} \rho_j\bigl([{\bm{x}(t)}]\bigr) =
\frac{d}{dt}\bigg|_{t=0}
\bigl(\bm{x}(t)^*\bm{P}_j\bm{x}(t)\bigr)=
\bm{z}^*\bm{P}_j\bm{x}+\bm{x}^*\bm{P}_j\bm{z}=
2\operatorname{Re}\bigl(\bm{z}^*\bm{P}_j\bm{x}\bigr).
\]
Using the Hermitian matrix $\bm{X}$ defined above, we can compactly rewrite this differential as
\[
    d\rho_j\big|_{[\bm{x}]}(\bm{v}) = \operatorname{Tr}(\bm{P}_j\bm{X})
\]
for every $j=1,\ldots,10$. Consequently, the differential of $F$ evaluated at $\bm{v}$, viewed as a vector in $\mathbb{R}^{10}$, is exactly
\[
    dF_{[\bm{x}]}(\bm{v}) = \cA_{\bm{B}}(\bm{X}).
\]
Suppose that the tangent vector $\bm{v}$ is in the kernel of $dG_{[\bm{x}]}$. Then $dG_{[\bm{x}]}(\bm{v})=\bm{0}$, which implies $dF_{[\bm{x}]}(\bm{v})=\bm{0}$ and thus $\cA_{\bm{B}}(\bm{X})=\bm{0}$. By Lemma~\ref{lem:kernel}, this forces $\bm{X}=\bm{0}$. Returning to the definition of $\bm{X}$, we have
\[
    \bm{X}\bm{x} = \bm{x}(\bm{z}^*\bm{x})+\bm{z}(\bm{x}^*\bm{x}) = \bm{0} + \bm{z}(1) = \bm{z}.
\]
Thus, $\bm{X}=\bm{0}$ necessarily forces $\bm{z}=\bm{0}$, which in turn implies $\bm{v}=dq_{\bm{x}}(\bm{0})=\bm{0}$. This proves that $dG_{[\bm{x}]}$ is injective, confirming that $G$ is a smooth immersion.
This completes the proof.
\end{proof}

Let
\[
    \pi_\nu:\nu\longrightarrow\mathbb{CP}^3
\]
denote the normal bundle associated with the embedding
$G:\mathbb{CP}^3\hookrightarrow\mathbb{E}_0$. Its fiber over a point
$[\bm{x}]\in\mathbb{CP}^3$ is defined by
\[
    \nu_{[\bm{x}]}
    =
    \left(dG_{[\bm{x}]}
    \bigl(T_{[\bm{x}]}\mathbb{CP}^3\bigr)\right)^\perp
    \subset\mathbb{E}_0,
\]
where the orthogonal complement is taken with respect to the standard Euclidean metric inherited from $\mathbb{R}^{10}$. Since $\dim_{\mathbb{R}}\mathbb{E}_0=9$ and $\dim_{\mathbb{R}}\mathbb{CP}^3=6$, the normal bundle has a real rank of
\begin{equation}\label{eq:normal-rank}
    \operatorname{rank}_{\mathbb{R}}\nu=9-6=3.
\end{equation}
Both $\mathbb{CP}^3$ and $\mathbb{E}_0$ possess natural orientations: $\mathbb{CP}^3$ is equipped with its canonical complex orientation, and $\mathbb{E}_0\cong\mathbb{R}^9$ carries the standard Euclidean orientation. The canonical Whitney sum splitting
\[
    T\mathbb{CP}^3\oplus\nu\cong\underline{\mathbb{R}}^{\,9}
\]
therefore uniquely determines an orientation on $\nu$. 

We now exploit the rank-one structure of the intensity coordinates. The Parseval identity $\sum_{j=1}^{10}\bm{b}_j\bm{b}_j^*=\bm{I}_4$ ensures that at least one measurement vector $\bm{b}_j$ is nonzero. Fix such an index $j$, and define the zero locus
\begin{equation}\label{eq:hyperplane}
   H_j
   =
   \{[\bm{x}]\in\mathbb{CP}^3:\bm{b}_j^*\bm{x}=0\}.
\end{equation}
Because $\bm{b}_j$ acts as a nonzero complex linear functional, its kernel constitutes a three-dimensional complex subspace of $\mathbb{C}^4$. The projectivization of this kernel is therefore precisely
\[
    H_j\cong\mathbb{CP}^2.
\]
Let $\bm{e}_j$ denote the $j$-th standard basis vector of $\mathbb{R}^{10}$, and construct the constant vector
\begin{equation}\label{eq:normal-vector}
   \bm{\gamma}_j=\bm{e}_j-\frac{1}{10}\bm{1}\in\mathbb{E}_0.
\end{equation}
The following proposition demonstrates that $\bm{\gamma}_j$ constitutes a nowhere-zero smooth section of the restricted normal bundle $\nu|_{H_j}$.

\begin{proposition}\label{prop:normal-section}
Let $\bm{B}=(\bm{b}_1,\ldots,\bm{b}_{10})$ be a Parseval frame possessing the phase retrieval property, let $G$ be defined as in \eqref{eq:centered-map}, and let $\nu$ be its associated normal bundle. For every index $j$ such that $\bm{b}_j\neq\bm{0}$, the subspace $H_j$ defined in \eqref{eq:hyperplane} is canonically isomorphic to $\mathbb{CP}^2$, and the vector $\bm{\gamma}_j$ in \eqref{eq:normal-vector} provides a nowhere-zero smooth section of the restricted normal bundle $\nu\vert_{H_j}$.
\end{proposition}

\begin{proof}
Consider a point $[\bm{x}]\in H_j$ and choose a unit representative $\bm{x} \in S^7$. As established in the proof of Proposition~\ref{prop:projective-embedding}, any tangent vector $\bm{v}\in T_{[\bm{x}]}\mathbb{CP}^3$ can be uniquely written as $\bm{v} = dq_{\bm{x}}(\bm{z})$ for some horizontal vector $\bm{z}\in \bm{x}^\perp_{\mathbb{C}}$. 
According to the definition of $H_j$ in \eqref{eq:hyperplane}, we have $\bm{b}_j^*\bm{x}=0$. Evaluating the differential $d\rho_j$ yields
\[
    d\rho_j\big|_{[\bm{x}]}(\bm{v})
    =
    2\operatorname{Re}\bigl((\bm{b}_j^*\bm{x})^*(\bm{b}_j^*\bm{z})\bigr)
    =
    0.
\]
Importantly, this vanishing property holds for \emph{all} tangent vectors to $\mathbb{CP}^3$ at $[\bm{x}]$, not merely those tangent to the subspace $H_j$. Moreover, the coordinate relation $\sum_{k=1}^{10}\rho_k=1$ immediately implies that
\[
    \sum_{k=1}^{10}d\rho_k\big|_{[\bm{x}]}(\bm{v})=0.
\]
Consequently, applying the definition of $G$ and evaluating the standard inner product on $\mathbb{E}_0$, we obtain
\[
    \left\langle\bm{\gamma}_j,dG_{[\bm{x}]}(\bm{v})\right\rangle
    =
    d\rho_j\big|_{[\bm{x}]}(\bm{v})
    -
    \frac{1}{10}\sum_{k=1}^{10}d\rho_k\big|_{[\bm{x}]}(\bm{v})
    =
    0.
\]
This confirms that $\bm{\gamma}_j$ is strictly orthogonal to the image of the tangent space, meaning $\bm{\gamma}_j \in \nu_{[\bm{x}]}$ for all $[\bm{x}]\in H_j$. Furthermore, this section is nowhere zero, as its Euclidean norm is constant and strictly positive:
\[
    \|\bm{\gamma}_j\|^2
    =
    \left\|\bm{e}_j-\frac{1}{10}\bm{1}\right\|^2
    =
    1-\frac{2}{10}+\frac{10}{100}
    =
    \frac{9}{10}.
\]
Thus, $\bm{\gamma}_j$ constitutes a valid nowhere-zero smooth section of $\nu|_{H_j}$.
\end{proof}

Proposition~\ref{prop:normal-section} immediately yields the following topological consequence.

\begin{corollary}\label{cor:normal-splitting}
Under the hypotheses of Proposition~\ref{prop:normal-section}, there
exists an oriented real rank-two vector bundle $\eta_j\to H_j$ such that
\begin{equation} \label{eq:nu}
    \nu|_{H_j}\cong\underline{\mathbb{R}}\oplus\eta_j.
\end{equation}
\end{corollary}

\begin{proof}
The nowhere-zero section $\bm{\gamma}_j$ globally spans a trivial oriented real line subbundle within $\nu|_{H_j}$. Taking the orthogonal complement of this line bundle with respect to the standard Euclidean metric on $\mathbb{E}_0$ yields a real rank-two complementary bundle $\eta_j$. Because both the total bundle $\nu|_{H_j}$ and the trivial line subbundle are oriented, their direct sum relationship canonically induces an orientation on $\eta_j$.
\end{proof}

We now compute the necessary characteristic classes. Following the cohomology ring structure detailed in Section~\ref{sec:notation}, $H^4(\mathbb{CP}^3;\mathbb{Z})$ is the cohomology group containing the first Pontryagin class of the tangent bundle. Recall the generator $h = c_1(\mathcal{O}_{\mathbb{CP}^3}(1))$ and the truncation relation $h^4=0$ defined in \eqref{eq:cohomology-ring}. Here, $\cO_{\mathbb{CP}^3}(1)$ represents the hyperplane line bundle, and $c_k$ denotes the $k$-th Chern class.  The Euler sequence for $\mathbb{CP}^3$ \cite{GriffithsHarris}  is the following short exact sequence of holomorphic vector bundles:
\[
    0\longrightarrow\cO_{\mathbb{CP}^3}
    \longrightarrow
    \cO_{\mathbb{CP}^3}(1)^{\oplus 4}
    \longrightarrow
    T^{1,0}\mathbb{CP}^3
    \longrightarrow 0,
\]
where $\cO_{\mathbb{CP}^3}$ is the trivial holomorphic line bundle, $\cO_{\mathbb{CP}^3}(1)^{\oplus 4}$ is the direct sum of four copies of $\cO_{\mathbb{CP}^3}(1)$, and $T^{1,0}\mathbb{CP}^3$ is the holomorphic tangent bundle of complex rank $3$.
For any short exact sequence of complex vector bundles $0\to E'\to E\to E''\to 0$, the Whitney sum formula \cite{MilnorStasheff} dictates that $c(E)=c(E')c(E'')$. Applying this to the Euler sequence, and utilizing $c(\cO_{\mathbb{CP}^3})=1$ along with $c(\cO_{\mathbb{CP}^3}(1))=1+h$, we obtain
\[
    c\bigl(T^{1,0}\mathbb{CP}^3\bigr)
    =
    \frac{
    c\bigl(\cO_{\mathbb{CP}^3}(1)^{\oplus 4}\bigr)}
    {c(\cO_{\mathbb{CP}^3})}
    =(1+h)^4.
\]
Since $h^4=0$, expanding this polynomial yields
\[
    c\bigl(T^{1,0}\mathbb{CP}^3\bigr)
    =
    1+4h+6h^2+4h^3.
\]
Consequently, the individual Chern classes are given by
\[
    c_1\bigl(T^{1,0}\mathbb{CP}^3\bigr)=4h,
    \qquad
    c_2\bigl(T^{1,0}\mathbb{CP}^3\bigr)=6h^2.
\]
For a real vector bundle $\xi$, its first Pontryagin class $p_1(\xi)$ resides in $H^4(X;\mathbb{Z})$. If $E$ is a complex vector bundle and $E_{\mathbb{R}}$ denotes its underlying real vector bundle by forgetting the complex structure, we have the standard relation
\[
    p_1(E_{\mathbb{R}})=c_1(E)^2-2c_2(E).
\]
This formula stems from the isomorphism $(E_{\mathbb{R}})_{\mathbb{C}} \cong E\oplus\overline{E}$ and the definition $p_1(E_{\mathbb{R}}) = -c_2\bigl((E_{\mathbb{R}})_{\mathbb{C}}\bigr)$, where $\overline{E}$ is the complex conjugate bundle. Since the real tangent bundle satisfies $T\mathbb{CP}^3 \cong \bigl(T^{1,0}\mathbb{CP}^3\bigr)_{\mathbb{R}}$, we compute
\[
    \begin{aligned}
    p_1(T\mathbb{CP}^3)
    &=
    c_1\bigl(T^{1,0}\mathbb{CP}^3\bigr)^2
    -
    2c_2\bigl(T^{1,0}\mathbb{CP}^3\bigr)\\
    &=(4h)^2-2(6h^2)
    =4h^2.
    \end{aligned}
\]
For the normal bundle $\nu$ of the embedding $G$, we have the stable relation
\[
    T\mathbb{CP}^3\oplus\nu\cong\underline{\mathbb{R}}^{\,9}.
\]
Generally, the Whitney sum formula for Pontryagin classes only holds modulo $2$-torsion. However, since the cohomology group $H^4(\mathbb{CP}^3;\mathbb{Z}) \cong \mathbb{Z}$ is torsion-free, the formula holds exactly over the integers. This gives
\[
    0=p_1(\underline{\mathbb{R}}^{\,9})
    =p_1(T\mathbb{CP}^3)+p_1(\nu).
\]
Substituting $p_1(T\mathbb{CP}^3)=4h^2$ yields
\begin{equation} \label{eq:4hsq}
    p_1(\nu)=-4h^2
    \qquad\text{in }H^4(\mathbb{CP}^3;\mathbb{Z}).
\end{equation}
Choosing a nonzero $\bm{b}_j$ as in \eqref{eq:hyperplane}, let $\iota_j:H_j\hookrightarrow\mathbb{CP}^3$ be the inclusion map, and define the pullback $\bar{h}=\iota_j^*h\in H^2(H_j;\mathbb{Z})$. Since $H_j\cong\mathbb{CP}^2$, the class $\bar{h}$ generates $H^2(H_j;\mathbb{Z})\cong\mathbb{Z}$, and $\bar{h}^2$ generates $H^4(H_j;\mathbb{Z})\cong\mathbb{Z}$. Restricting \eqref{eq:4hsq} to $H_j$ gives
\[
    p_1(\nu|_{H_j})=-4\bar{h}^2.
\]
On the other hand, Corollary~\ref{cor:normal-splitting} establishes the splitting
\[
    \nu|_{H_j}\cong\underline{\mathbb{R}}\oplus\eta_j.
\]
Because the first Pontryagin class is additive under direct sums
 %(again, strictly over $\mathbb{Z}$ since $H^4(H_j; \mathbb{Z})$ is torsion-free) 
 and the trivial line bundle has vanishing characteristic classes, we have
\[
    p_1(\nu|_{H_j})=p_1(\eta_j).
\]
For an oriented real rank-two vector bundle $\eta$, its structure group reduces to $\mathrm{SO}(2)$, and its first Pontryagin class equals the square of its Euler class $e(\eta)\in H^2(H_j;\mathbb{Z})$:
\[
    p_1(\eta)=e(\eta)^2,
\]
see \cite[Chapter~15]{MilnorStasheff}.
Since $H^2(H_j;\mathbb{Z})$ is generated by $\bar{h}$, there exists an integer $k\in\mathbb{Z}$ such that
\[
    e(\eta_j)=k\bar{h}.
\]
Consequently,
\[
    p_1(\nu|_{H_j})
    =
    p_1(\eta_j)
    =
    e(\eta_j)^2
    =
    k^2\bar{h}^2.
\]
Equating the two expressions for $p_1(\nu|_{H_j})$ and noting that $\bar{h}^2$ is a generator of $H^4(H_j;\mathbb{Z})$ yields the algebraic relation
\[
    k^2=-4.
\]
This is clearly impossible for any real integer $k\in\mathbb{Z}$. This  contradiction proves that no family of exactly ten vectors can possess the phase retrieval property in $\mathbb{C}^4$. This completes the proof.

\section{Discussion}\label{sec:discussion}
In this paper, we have established that no family of ten measurement vectors can perform phase retrieval in $\mathbb{C}^4$. The proof proceeds by constructing a smooth projective embedding
\[
    \mathbb{CP}^3\hookrightarrow\mathbb{R}^9.
\]
By exploiting the special zero loci of the rank-one intensity coordinates, we force a Whitney sum splitting of the rank-three normal bundle over a canonically embedded copy of $\mathbb{CP}^2$. The stable tangent-normal relation rigidly fixes the restricted first Pontryagin class to be $-4\bar{h}^2$, whereas the geometric splitting forces it to be a perfect square $k^2\bar{h}^2$. This topological obstruction leads directly to a contradiction. Combined with Vinzant's explicit $11$-vector construction, this conclusively proves that the exact minimum number of measurements required for phase retrieval in $\mathbb{C}^4$ is $11$. 

Our main theorem addresses the minimal measurement number problem in $\mathbb{C}^4$. For higher dimensions $M\ge 6$ that do not take the form $M = 2^k + 1$, $k \in \mathbb{N}_0$, determining the exact minimal measurement number remains a challenging open problem. We anticipate that extending our topological framework to these higher-dimensional spaces will serve as an exciting and promising direction for future research.

\section*{Acknowledgments}
M. Huang thanks Zhenxiao Xie for his careful reading of the manuscript and valuable feedback. This work was supported in part by the Beijing Natural Science Foundation (Grant No. 1262013), the National Key Research and Development Program of China (Grant No. 2025YFA1016902), and the National Natural Science Foundation of China (Grant No. 12201022).


\begin{thebibliography}{99}

\bibitem{BCE06}
R. Balan, P. G. Casazza, D. Edidin. On signal reconstruction without phase. \emph{Appl. Comput. Harmon. Anal.}, 2006, 20(3): 345--356.

\bibitem{BCMN14}
A. S. Bandeira, J. Cahill, D. G. Mixon, A. A. Nelson. Saving phase: Injectivity and stability for phase retrieval. \emph{Appl. Comput. Harmon. Anal.}, 2014, 37(1): 106--125.

\bibitem{BH1}
B. G. Bodmann, N. Hammen. Stable phase retrieval with low-redundancy frames. \emph{Adv. Comput. Math.}, 2015, 41(2): 317--331.

\bibitem{BH2}
B. G. Bodmann, N. Hammen. Algorithms and error bounds for noisy phase retrieval with low-redundancy frames. \emph{Appl. Comput. Harmon. Anal.}, 2017, 43(3): 482--503.

\bibitem{BCNT}
S. Botelho-Andrade, P. G. Casazza, H. Van Nguyen, J. C. Tremain. Phase retrieval versus phaseless reconstruction. \emph{J. Math. Anal. Appl.}, 2016, 436(1): 131--137.

\bibitem{Tcai}
T. Cai, A. Zhang. ROP: Matrix recovery via rank-one projections. \emph{Ann. Stat.}, 2015, 43(1): 102--138.

\bibitem{TTCai}
T. T. Cai, X. Li, Z. Ma. Optimal rates of convergence for noisy sparse phase retrieval via thresholded Wirtinger flow. \emph{Ann. Stat.}, 2016, 44(5): 2221--2251.

\bibitem{cai2021solving}
J. F. Cai, M. Huang, D. Li, Y. Wang. Solving phase retrieval with random initial guess is nearly as good as by spectral initialization. \emph{Appl. Comput. Harmon. Anal.}, 2022, 58: 60--84.

\bibitem{Phaseliftn}
E. J. Cand\`es, X. Li. Solving quadratic equations via PhaseLift when there are about as many equations as unknowns. \emph{Found. Comput. Math.}, 2014, 14(5): 1017--1026.

\bibitem{WF}
E. J. Cand\`es, X. Li, M. Soltanolkotabi. Phase retrieval via Wirtinger flow: Theory and algorithms. \emph{IEEE Trans. Inf. Theory}, 2015, 61(4): 1985--2007.

\bibitem{phaselift}
E. J. Cand\`es, T. Strohmer, V. Voroninski. Phaselift: Exact and stable signal recovery from magnitude measurements via convex programming. \emph{Commun. Pure Appl. Math.}, 2013, 66(8): 1241--1274.



\bibitem{Carmeli16}
C. Carmeli, T. Heinosaari, M.  Kech,  J. Schultz, A. Toigo. Stable pure state quantum tomography from five orthonormal bases. \emph{Europhysics Letters}, 2016, 115(3): 30001.


\bibitem{Carmeli15}
C. Carmeli, T. Heinosaari, J. Schultz, A. Toigo. How many orthonormal bases are needed to distinguish all pure quantum states? \emph{Eur. Phys. J. D}, 2015, 69: 179.



\bibitem{chai2010array}
A. Chai, M. Moscoso, G. Papanicolaou. Array imaging using intensity-only measurements. \emph{Inverse Probl.}, 2011, 27(1): 015005.

\bibitem{ChenFan}
Y. Chen, Y. Chi, J. Fan, C. Ma. Gradient descent with random initialization: Fast global convergence for nonconvex phase retrieval. \emph{Math. Program.}, 2019, 176: 5--37.

\bibitem{CEHV15}
A. Conca, D. Edidin, M. Hering, C. Vinzant. An algebraic characterization of injectivity in phase retrieval. \emph{Appl. Comput. Harmon. Anal.}, 2015, 38(2): 346--356.

\bibitem{fienup1987phase}
J. C. Dainty, J. R. Fienup. Phase retrieval and image reconstruction for astronomy. \emph{Image Recovery: Theory and Appl.}, 1987, 231: 275.

\bibitem{Duchi}
J. C. Duchi, F. Ruan. Solving (most) of a set of quadratic equalities: Composite optimization for robust phase retrieval. \emph{Inf. Inference}, 2019, 8(3): 471--529.

\bibitem{F04}
J. Finkelstein. Pure-state informationally complete and `really' complete measurements. \emph{Phys. Rev. A}, 2004, 70(5): 052107.

\bibitem{FSC05}
S. T. Flammia, A. Silberfarb, C. M. Caves. Minimal informationally complete measurements for pure states. \emph{Found. Phys.}, 2005, 35(12): 1985--2006.


\bibitem{GriffithsHarris}
P. Griffiths, J. Harris. Principles of Algebraic Geometry. \emph{John Wiley and Sons}, 2014.

\bibitem{GoyenecheEtAl2015}
D. Goyeneche, G. Ca\~nas, S. Etcheverry, E. S. G\'omez, G. B. Xavier, G. Lima, A. Delgado. Five measurement bases determine pure quantum states on any dimension. \emph{Phys. Rev. Lett.}, 2015, 115(9): 090401.

\bibitem{GKR}
P. Grohs, S. Koppensteiner, M. Rathmair. Phase Retrieval: Uniqueness and Stability. \emph{SIAM Rev.}, 2020, 62(2): 301--350.

\bibitem{harrison1993phase}
R. W. Harrison. Phase problem in crystallography. \emph{J. Opt. Soc. Am. A}, 1993, 10(5): 1046--1055.

\bibitem{Heinosaari13}
T. Heinosaari, L. Mazzarella, M. M. Wolf. Quantum tomography under prior information. \emph{Commun. Math. Phys.}, 2013, 318(2): 355--374; corrected arXiv version: arXiv:1109.5478v3 (2018).

\bibitem{huangTIT}
M. Huang. Near-Quadratic Convergence of the Gauss--Newton Method for Complex Phase Retrieval. \emph{IEEE Trans. Inf. Theory}, 2026, 72(1): 222--245.

\bibitem{huangsiam}
M. Huang, Y. Wang. Linear convergence of randomized Kaczmarz method for solving complex-valued phaseless equations. \emph{SIAM J. Imaging Sci.}, 2022, 15(2): 989--1016.

\bibitem{Lee09}
J. M. Lee. Manifolds and Differential Geometry. \emph{Grad. Stud. Math.}, 2009, 107: American Mathematical Society.

\bibitem{Lee12}
J. M. Lee. Introduction to Smooth Manifolds (2nd ed.). \emph{Graduate Texts in Mathematics, Springer, New York}, 2013, 218.

\bibitem{Li2026}
Z. Li. On injectivity of phase retrieval. \emph{arXiv preprint}, 2026, arXiv:2606.17922.

\bibitem{miao2008extending}
J. Miao, T. Ishikawa, Q. Shen, T. Earnest. Extending X-ray crystallography to allow the imaging of noncrystalline materials, cells, and single protein complexes. \emph{Annu. Rev. Phys. Chem.}, 2008, 59: 387--410.

\bibitem{Milgram}
R. J. Milgram. Immersing projective spaces. \emph{Ann. Math.}, 1967, 85(3): 473--482.

\bibitem{millane1990phase}
R. P. Millane. Phase retrieval in crystallography and optics. \emph{J. Opt. Soc. Am. A}, 1990, 7(3): 394--411.

\bibitem{MilnorStasheff}
J. W. Milnor, J. D. Stasheff. Characteristic Classes. \emph{Ann. Math. Stud.}, 1974, 76: Princeton University Press.

\bibitem{shechtman2015phase}
Y. Shechtman, Y. C. Eldar, O. Cohen, H. N. Chapman, J. Miao, M. Segev. Phase retrieval with application to optical imaging: a contemporary overview. \emph{IEEE Signal Process. Mag.}, 2015, 32(3): 87--109.



\bibitem{SunLL}
L. L. Sun, S. Yu,  Z.-B. Chen.  Minimal determination of a pure qutrit state and four-measurement protocol for pure qudit state.  \emph{J. Phys. A Math. Theor.}, 2020, 53, 075305.

\bibitem{turstregion}
J. Sun, Q. Qu, J. Wright. A geometric analysis of phase retrieval. \emph{Found. Comput. Math.}, 2018, 18(5): 1131--1198.

\bibitem{tan2019phase}
Y. S. Tan, R. Vershynin. Phase retrieval via randomized kaczmarz: theoretical guarantees. \emph{Inf. Inference}, 2019, 8(1): 97--123.

\bibitem{Vinzant15}
C. Vinzant. A small frame and a certificate of its injectivity. \emph{Proc. SampTA}, 2015, 197--200.

\bibitem{Waldspurger2015}
I. Waldspurger, A. d'Aspremont, S. Mallat. Phase recovery, maxcut and complex semidefinite programming. \emph{Math. Program.}, 2015, 149(1-2): 47--81.

\bibitem{waldspurger2018phase}
I. Waldspurger. Phase retrieval with random gaussian sensing vectors by alternating projections. \emph{IEEE Trans. Inf. Theory}, 2018, 64(5): 3301--3312.

\bibitem{walther1963question}
A. Walther. The question of phase retrieval in optics. \emph{J. Mod. Opt.}, 1963, 10(1): 41--49.

\bibitem{WangShang2018}
Y. Wang, Y. Shang. Pure state `really' informationally complete with rank-1 POVM. \emph{Quantum Inf. Process.}, 2018, 17(3): 51.

\bibitem{wangxu}
Y. Wang, Z. Xu. Generalized phase retrieval: measurement number, matrix recovery and beyond. \emph{Appl. Comput. Harmon. Anal.}, 2019, 47(2): 423--446.

\end{thebibliography}
\end{document}